\documentclass{article}

\usepackage{amsmath}
\usepackage{amssymb}
\usepackage{latexsym}
\usepackage{verbatim}
\usepackage{subfigure}
\usepackage[final]{graphicx}
\usepackage{psfrag}

\usepackage{amsthm}

\newtheorem{theorem}{Theorem}

\newtheorem{lemma}{Lemma}

\newtheorem{remark}{Remark}

{\begin{list}               %    with flush left bullets.
    {$\bullet$ \hfill}{
        \setlength{\leftmargin}{\parindent}
        \setlength{\parsep}{0.04\baselineskip}
        \setlength{\itemsep}{0.5\parsep}
        \setlength{\labelwidth}{\leftmargin}
        \setlength{\labelsep}{0em}}
    }
{\end{list}}

\providecommand{\eref}[1]{\eqref{eq:#1}}  % call \eqref from amstex
\providecommand{\cref}[1]{Chapter~\ref{chap:#1}}
\providecommand{\sref}[1]{Section~\ref{sec:#1}}
\providecommand{\fref}[1]{Figure~\ref{fig:#1}}

\providecommand{\R}{\ensuremath{\mathbb{R}}}
\providecommand{\C}{\ensuremath{\mathbb{C}}}

\providecommand{\abs}[1]{\lvert#1\rvert}
\providecommand{\norm}[1]{\lVert#1\rVert}

\providecommand{\set}[1]{\left\{#1\right\}}

\providecommand{\bydef}{\overset{\text{def}}{=}}

\renewcommand{\vec}[1]{\ensuremath{\boldsymbol{#1}}}
\providecommand{\mat}[1]{\ensuremath{\boldsymbol{#1}}}

\providecommand{\mA}{\mat{A}} \providecommand{\mB}{\mat{B}}
 
\providecommand{\mD}{\mat{D}}
\providecommand{\mH}{\mat{H}}
\providecommand{\mI}{\mat{I}}  
 \providecommand{\mL}{\mat{L}} 
\providecommand{\mM}{\mat{M}}  
\providecommand{\mQ}{\mat{Q}} \providecommand{\mR}{\mat{R}}
 \providecommand{\mU}{\mat{U}} 
\providecommand{\mV}{\mat{V}}

\providecommand{\mZ}{\mat{Z}}
\providecommand{\mSigma}{\mat{\Sigma}}
 \providecommand{\mG}{\mat{G}}

\providecommand{\va}{\vec{a}} 
 \providecommand{\ve}{\vec{e}}
\providecommand{\vg}{\vec{g}}

 \providecommand{\vp}{\vec{p}}

\providecommand{\vu}{\vec{u}} \providecommand{\vw}{\vec{w}}
\providecommand{\vx}{\vec{x}} \providecommand{\vy}{\vec{y}}
 
 \providecommand{\vzero}{\vec{0}}
 
\providecommand{\vxi}{\vec{\xi}}
\providecommand{\vv}{\vec{v}}

\usepackage[margin=1in]{geometry}

\usepackage{cite}

\usepackage{color}

\usepackage{booktabs,multirow}

\usepackage{enumitem}
\usepackage{pmat}
\usepackage{commath}
\usepackage{mathtools}
\mathtoolsset{showonlyrefs}

\usepackage{algorithm,algpseudocode}

\usepackage{hyperref}

\usepackage{esdiff}

\usepackage{pgf,tikz,psfrag}

\usepackage{dsfont}

\newcommand{\GE}{\text{Ginibre}}
\newcommand{\HE}{\text{Haar}}
\newcommand{\Ortho}{\mathbb{O}}
\newcommand{\Unitary}{\mathbb{U}}

\providecommand{\eqd}{\overset{d}{=}}

\DeclareMathOperator{\diag}{diag}

\DeclareMathOperator{\sign}{sign}

\newcommand*{\tran}{^{\mkern-1.5mu\mathsf{T}}}
\newcommand*{\herm}{^\ast}

\providecommand{\vxi}{\vec{\xi}}

\providecommand{\vbeta}{\vec{\beta}}

\title{Householder Dice: A Matrix-Free Algorithm for Simulating Dynamics on Gaussian and Random Orthogonal Ensembles}

\author{%
Yue M. Lu
\thanks{Y. M. Lu is with the John A. Paulson School of Engineering and Applied Sciences, Harvard University, Cambridge, MA 02138, USA (e-mail:  \href{mailto:yuelu@seas.harvard.edu}{yuelu@seas.harvard.edu}). The initial part of this work was done during his sabbatical at the \'{E}cole normale sup\'{e}rieure (ENS) in Paris, France in Fall 2019. He thanks colleagues at the ENS for their hospitality and stimulating discussions. This work was supported by the Harvard FAS Dean's Fund for Promising Scholarship, by the chaire CFM-ENS ``Science des donnees'', and by the US National Science Foundation under grants CCF-1718698 and CCF-1910410.}%
}
\date{}

\begin{document}

%\markboth{}{Lu: Householder Dice}

\maketitle

\begin{abstract}
This paper proposes a new algorithm, named Householder Dice (HD), for simulating dynamics on dense random matrix ensembles with translation-invariant properties. Examples include the Gaussian ensemble, the Haar-distributed random orthogonal ensemble, and their complex-valued counterparts. A ``direct'' approach to the simulation, where one first generates a dense $n \times n$ matrix from the ensemble, requires at least $\mathcal{O}(n^2)$ resource in space and time. The HD algorithm overcomes this $\mathcal{O}(n^2)$ bottleneck by using the principle of deferred decisions: rather than fixing the entire random matrix in advance, it lets the randomness unfold with the dynamics. At the heart of this matrix-free algorithm is an adaptive and recursive construction of (random) Householder reflectors. These orthogonal transformations exploit the group symmetry of the matrix ensembles, while simultaneously maintaining the statistical correlations induced by the dynamics. The memory and computation costs of the HD algorithm are $\mathcal{O}(nT)$ and $\mathcal{O}(nT^2)$, respectively, with $T$ being the number of iterations. When $T \ll n$, which is nearly always the case in practice, the new algorithm leads to significant reductions in runtime and memory footprint. Numerical results demonstrate the promise of the HD algorithm as a new computational tool in the study of high-dimensional random systems.
\end{abstract}

%\begin{IEEEkeywords}
%Haar distribution, Householder reflection, random matrices
%\end{IEEEkeywords}

\section{Introduction}
\label{sec:intro}

To do research involving large random systems, one must make a habit of experimenting on the computer. Indeed, computer simulations help verify theoretical results and provide new insights, not to mention that they can also be incredibly fun.\ For many problems in statistical learning, random matrix theory, and statistical physics, the simulations that one encounters are often given as an iterative process in the form of
\begin{equation}\label{eq:dynamics}
\vx_{t+1} = f_t(\mM_t \vx_t, \vx_t, \vx_{t-1}, \ldots, \vx_{t-d}), \qquad \text{for } 1 \le t \le T.
\end{equation}
Here, $\mM_t$ is either $\mQ$ or $\mQ\tran$, where $\mQ$ is a random matrix; $f_t(\cdot)$ denotes some general vector-valued function that maps $\mM_t \vx_t$ and a few previous iteration vectors $\set{\vx_{t-i}}_{0 \le i \le d}$ to the next one $\vx_{t+1}$; and $T$ is the total number of iterations.

With suitable definitions of the mappings $f_t(\cdot)$, the formulation in \eref{dynamics} includes many well-known algorithms as its special cases. A classical example is to use iterative methods \cite{Stewart2002eigen} to compute the extremal eigenvalues/eigenvectors of a (spiked) random matrix \cite{BaikGP2005bbp,Benaych-GeorgesN2011spike}. Other examples include approximate message passing on dense random graphs \cite{Bolthausen2014iterative,BayatiM2011dynamics,OpperCW2016ising,RanganSF2019vamp,Fan2020amp}, and gradient descent algorithms for solving learning and estimation problems with random design \cite{CandesLS2015wirtinger,Goldt2019dynamics}. In this paper, we show that all of these algorithms can be simulated by an efficient \emph{matrix-free} scheme, if the random matrix $\mQ$ is drawn from an ensemble with translation-invariant properties. Examples of such ensembles include the i.i.d. Gaussian (i.e. the rectangular Ginibre) ensemble, the Haar-distributed random orthogonal ensemble, the Gaussian orthogonal ensemble, and their complex-valued counterparts.

What is wrong with the standard way of simulating \eref{dynamics}, where we first draw a sample $\mQ$ from the matrix ensemble and then carry through the iterations? This direct approach is straightforward to implement, but it cannot handle large dimensions. To see this, suppose that $Q \in \R^{m \times n}$ with $m \asymp n$. We shall also assume that the computational cost of the nonlinear mapping $f_t(\cdot)$ is $\mathcal{O}(n)$. It follows that, at each iteration of \eref{dynamics}, most of the computation is spent on the matrix-vector multiplication $\mM_t \vx_t$, at a cost of $\mathcal{O}(n^2)$ work. It is not at all obvious that one can do much better: Merely generating an $n\times n$ Gaussian matrix already requires $O(n^2)$ resource in computation and storage. When $n$ is large, $n^2$ is huge. In practice, this $\mathcal{O}(n^2)$ bottleneck means that one cannot simulate \eref{dynamics} at a dimension much larger than $n = 10^4$ on a standard computer (in a reasonable amount of time). However, there are many occasions, especially in the study of high-dimensional random systems, where one does wish to simulate large random matrices. A common workaround is to choose a moderate dimension (\emph{e.g.}, $n = 1000$), repeat the simulation over many independent trials, and then average the results to reduce statistical fluctuations. In addition to having to spend extra time on the repeated trials, this strategy can still suffer from strong finite size effects, making it a poor approximation of the true high-dimensional behavior of the underlying random systems. (An example is given in \sref{spectral} to illustrate this issue.)

In this paper, we propose a new algorithm, named \emph{Householder Dice} (HD), for simulating the dynamics in \eref{dynamics} on the Gaussian, Haar, and other related random matrix ensembles. Our new approach is \emph{statistically-equivalent} to the direct approach discussed above, but the memory and computation costs of the HD algorithm are $\mathcal{O}(nT)$ and $\mathcal{O}(nT^2)$, respectively, where $T$ is the number of iterations. In many problems, $T$ is much smaller than $n$. Typically, $T = \mathcal{O}(\text{polylog}(n))$.  In such cases, the new algorithm leads to significant reductions in runtime and memory footprint. In the numerical examples presented in \sref{examples}, we show that the crossover value of $n$ at which the HD algorithm outperforms the direct approach can be as low as $500$. The speedup becomes orders of magnitude greater for $n \ge 10^4$. Moreover, the HD algorithm expands the limits of what could be done on standard computers by making it tractable to perform dense random matrix experiments in dimensions as large as $n = 10^7$.

The basic idea of the HD algorithm follows the so-called principle of deferred decisions \cite{MitzenmacherU2005probability}. Intuitively, each iteration of \eref{dynamics} only probes $\mQ$ in a one-dimensional space spanned by $\vx_t$. Thus, if the total number of iterations $T \ll n$, we only need to expose the randomness of $\mQ$ over a few low-dimensional subspaces. It is then clearly wasteful to fix and store in memory the full matrix in advance. The situation is analogous to that of simulating a simple random walk for $T$ steps. We can let the random choices gradually unfold with the progress of the walk, fixing only the randomness that must be revealed at any given step. The challenge in our problem though is that the dynamics in \eref{dynamics} can create a complicated dependence structure between the random matrix $\mQ$ and the iteration vectors $\vx_t, \vx_{t-1} \ldots, \vx_0$. Nevertheless, we show that this dependence structure can be exactly accounted for by an adaptive and recursive construction of (random) Householder reflectors \cite{Householder1958reflector,TrefethenBau1997numerical} which exploit the inherent group symmetry of the matrix ensembles. 

Using Householder reflectors to speed up random matrix experiments is not a new idea. It is well-known \cite{Stewart1980ortho,Mezzadri2007generating} that a Haar-distributed random orthogonal matrix can be factorized as a product of Householder reflectors. This leads to an efficient way of generating a random orthogonal matrix with $\mathcal{O}(n^2)$ operations (rather than the $\mathcal{O}(n^3)$ cost associated with a full QR decomposition on a Gaussian matrix). Householder reflectors have also been applied to reduce a Gaussian matrix to a particularly simple random bidiagonal form \cite{Silverstein1985smallest,Edelman1989thesis}. This clever factorization leads to an $\mathcal{O}(n^2)$ algorithm for simulating the spectrum densities of Gaussian and Wishart matrices. (Recall that a standard eigenvalue decomposition on a dense matrix requires $\mathcal{O}(n^3)$ work in practice.) The proposed HD algorithm differs from the previous work in that it is a truly \emph{matrix-free} construction. With the progress of the dynamics, it gradually builds a recursive set of (random) Householder reflectors based on the current iteration vector $\vx_t$ and the history of the iterations up to this point. This adaptive, ``on-the-fly'' construction is essential for us to capture the correlation structures generated by the dynamics without fixing the matrix in advance.

The rest of the paper is organized as follows. We first present in \sref{examples} a few motivating examples to showcase the applications of the HD algorithm. \sref{main} contains the main technical results of this paper. After a brief review of the basic properties of the Haar measure (on classical matrix groups) and Householder reflectors, we present the construction of the proposed algorithm for the Gaussian and random orthogonal ensembles. Theorems~\ref{thm:Gaussian} and \ref{thm:Haar} establish the statistical equivalence of the HD algorithm and the direct approach to simulating \eref{dynamics}. Generalizations to complex-valued and other related ensembles are discussed in \sref{complex}. We conclude the paper in \sref{conclusion}.

%Simulations can also be incredibly fun and addictive, as nothing beats a numerical curve matching one's analytical prediction. 

%
%To many, this author included, computer simulations also provide a source of joy. The delight in seeing a numerical curve matching a (hard-fought) theoretical prediction is what makes doing research in the aforementioned fields so addictive. 

\section{Numerical Examples}
\label{sec:examples}

Before delving into technical details, it is helpful to go through a few motivating applications that show how the HD algorithm can significantly speed up the simulation tasks.\footnote{All of the numerical experiments presented in this section have been done in Julia \cite{BezansonEKS2017julia}. The source code implementing the HD algorithm is available online at \url{https://github.com/yuelusip/HouseholderDice}.}

\subsection{Lasso with Random Designs}

In the first example, we consider the simulation of the lasso estimator widely used in statistics and machine learning. The goal is to estimate a sparse vector $\vbeta^\ast \in \R^n$ from its noisy linear observation given by
\[
\vy = \mQ \vbeta^\ast + \vw,
\]
where $\mQ \in \R^{m \times n}$ is a design (or covariate) matrix, and $\vw \sim \mathcal{N}(0, \sigma_w^2 \mI)$ denotes the noise in $\vy$. The lasso estimator is formulated as an optimization problem
\begin{equation}\label{eq:lasso}
\widehat{\vbeta} = \underset{\vbeta}{\arg\min} \ \frac{1}{2} \norm{\vy - \mQ \vbeta}^2 + \lambda \norm{\vbeta}_1,
\end{equation}
where $\widehat{\vbeta}$ is an estimate of $\vbeta^\ast$ and $\lambda > 0$ is a regularization parameter.

A popular method for solving \eref{lasso} is the iterative soft-thresholding algorithm (ISTA) \cite{BeckT2009fista}:
\begin{equation}\label{eq:ista}
\vx_{t+1} = \eta_{\lambda \tau}[\vx_t + \tau \mQ\tran (\vy - \mQ \vx_t)], \qquad 0 \le t < T,
\end{equation}
where $\tau > 0$ denotes the step size and $\eta_{\lambda \tau}(x) = \sign(x) \max\set{\abs{x} - \lambda \tau, 0}$ is an element-wise soft-thresholding operator. In many theoretical studies of lasso, one assumes that the design matrix is random with i.i.d. normal entries, \emph{i.e.} $Q_{ij} \overset{\text{i.i.d.}}{\sim} \mathcal{N}(0, \tfrac{1}{m})$. In this case, ISTA is an iterative process on a Gaussian matrix $\mQ$ and its transpose. With some change of variables, we can rewrite \eref{ista} as a special case of the general dynamics given in \eref{dynamics}, with one iteration of \eref{ista} mapped to two iterations of \eref{dynamics}.

%\[
%\vx^{(0)} = \vbeta^\ast, \quad \vx^{(1)} = \vy = \mQ \vx^{(0)} + \vw, \quad \vx^{(2)} = \vbeta^{(0)}
%\]
%and starting from $t = 3$,
%\[
%\vx^{(t+1)} = \mQ \vx^{(t)} - \vx^{(1)}
%\]

We simulate the ISTA dynamics using both the proposed HD algorithm and the direct simulation approach that fixes the Gaussian matrix $\mQ$ in advance. In our experiments, the target sparse vector $\vbeta^\ast$ has i.i.d. entries drawn from the Bernoulli-Gaussian prior 
\[
\beta^\ast_i \sim \rho \delta(\beta) + (1-\rho) \frac{1}{\sqrt{2\pi \sigma_s^2}} \exp\Big\{-\frac{\beta^2}{2\sigma_s^2}\Big\},
\]
where $0 < \rho < 1$ and $\sigma_s > 0$ are two constants. The design matrix $\mQ$ is of size $m \times n$ with $m = \lfloor n/2 \rfloor$. 

\begin{figure}[t]
\centering
\subfigure[]{\label{fig:ist:1}
	\includegraphics[width=0.47\linewidth]{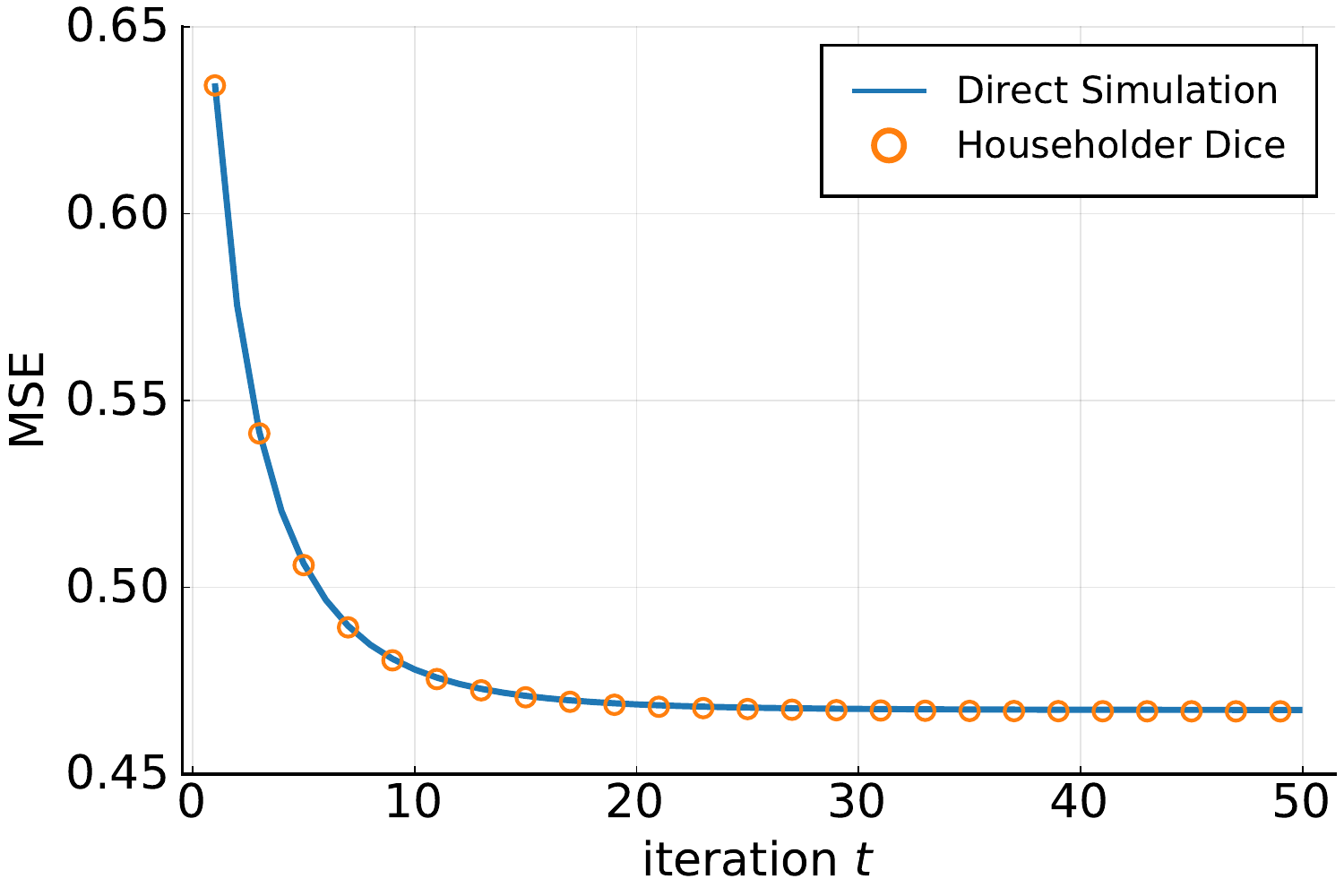}}
	\hspace{2ex}
	\subfigure[]{\label{fig:ist:2}
	\includegraphics[width=0.47\linewidth]{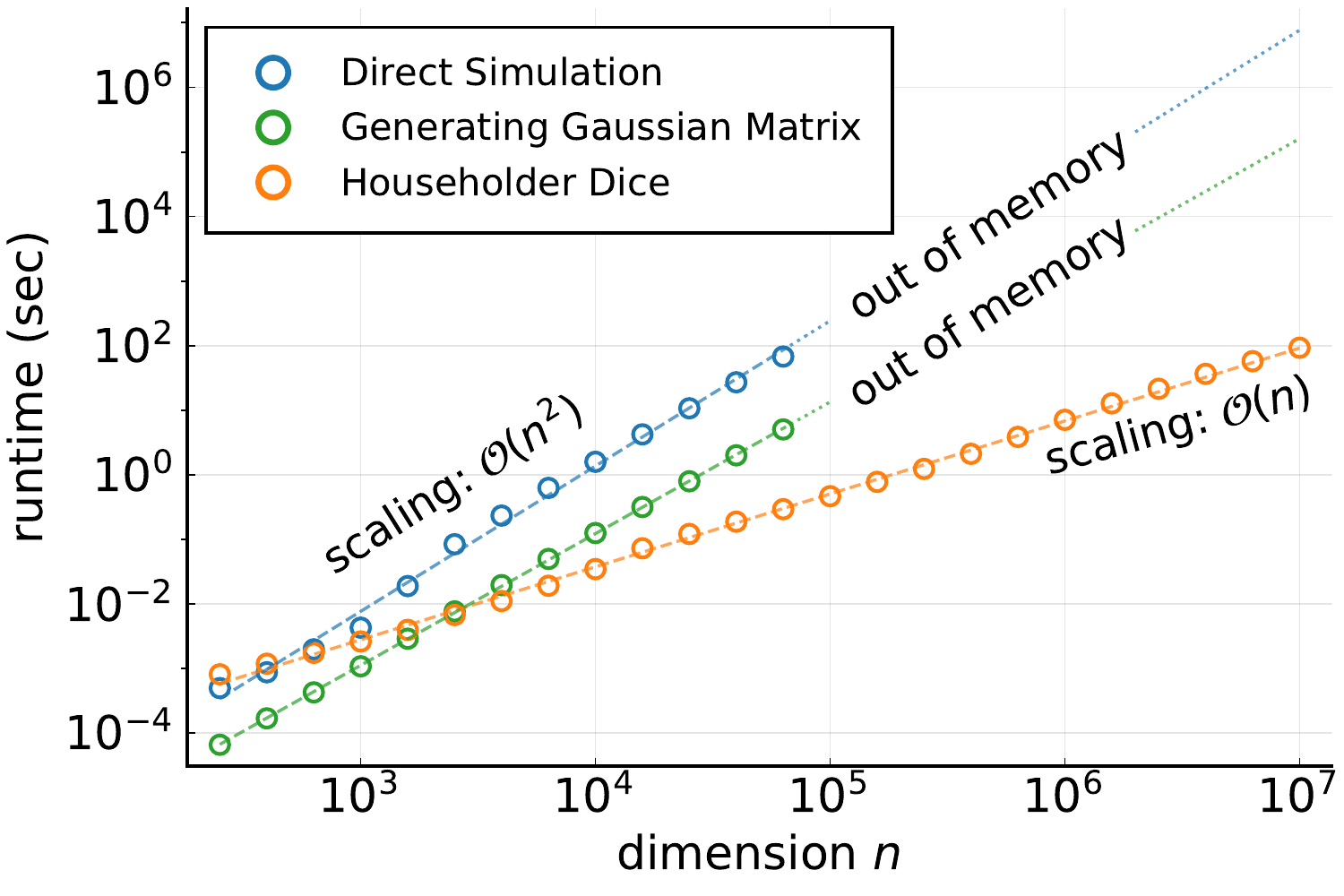}}
	\caption{Simulating the ISTA dynamics \eref{ista} using two approaches: the standard approach where the random matrix $\mQ$ is generated in advance, and the proposed HD algorithm. (a) The time-varying MSE averaged over $10^5$ independent trials, with the results from the two approaches matching. (b) Runtime versus the matrix dimension $n$, shown in log-log scale. In all the experiments, the parameters are set to $T=50$, $\lambda=2$, $\tau=0.3$, $\rho=0.2$, $\sigma_s = 2$ and $\sigma_w = 0.1$.}
\end{figure}

\fref{ist:1} shows the mean-squared error (MSE) $e^{(t)} \bydef \frac{1}{n} \norm{\vx_t - \vbeta^\ast}^2$ at each iteration of the dynamics, obtained by averaging over $10^5$ independent trials. The dimension here is $n = 1000$. The results from the HD algorithm (the red circles in the figure) match those from the standard approach (the blue line). This is expected, since Householder Dice is designed to be statistically equivalent to the direct approach. However, the two simulation approaches behave very differently in runtime and memory footprint, as shown in \fref{ist:2}. When we increase the dimension $n$, the runtime of the standard approach exhibits a quadratic growth rate $\mathcal{O}(n^2)$, whereas the runtime of the HD algorithm scales linearly with $n$. For comparison, we also plot in the figure the runtime for merely generating an i.i.d. Gaussian matrix $\mQ$ of size $m \times n$. 

For small dimensions ($250 \le n < 500$), the HD algorithm takes slightly more time than the direct approach, likely due to the additional overhead in implementing the former. Starting from $n \ge 500$, it becomes the more efficient choice. In fact, for $n \ge 2500$, the HD algorithm can simulate the ISTA dynamics (for 50 iterations) in less time than it takes to generate the Gaussian matrix. For dimensions beyond $n =10^5$, Householder Dice becomes the only feasible method, as implementing the direct approach would require more memory than available on the test computer (equipped with 32 GB of RAM). Finally, for $n = 10^7$, the runtime for the HD algorithm is 92 seconds, whereas by extrapolation the direct approach would have taken $7.7 \times 10^6$ seconds (approximately 89 days).

\subsection{Spectral Method for Generalized Linear Models}
\label{sec:spectral}

In the second example, we consider a spectral method  \cite{Li1992hessian,Netrapalli2013phase,LuL2020spectral,MondelliM2017spectral} with applications in signal estimation and exploratory data analysis. Let $\vxi$ be an unknown vector in $\R^n$ and $\set{\va_i}_{1 \le i \le m}$ a set of sensing vectors. We seek to estimate $\vxi$ from a number of generalized linear measurements $\set{y_i = f(\va_i\tran \vxi)}_{1 \le i \le m}$, where $f(\cdot)$ is some function modeling the acquisition process. The spectral method works as follows. Let
\begin{equation}\label{eq:D_mtx}
\mD \bydef \frac{1}{m} \mA \diag\set{y_1, \ldots, y_m} \mA\tran,
\end{equation}
where $\mA = [\va_1, \va_2, \ldots, \va_m]$ is a matrix whose columns are the sensing vectors. Denote by $\vx_1$ a normalized eigenvector associated with the largest eigenvalue of $\mD$. This vector $\vx_1$ is then our estimate of $\vxi$, up to a scaling factor. The performance of the spectral method is usually given in terms of the squared cosine similarity $\rho(\vxi, \vx_1) = \frac{(\vxi\tran \vx_1)^2}{\norm{\vxi}^2 \norm{\vx_1}^2}$.

Asymptotic limits of $\rho(\vxi, \vx_1)$ have been derived for the cases where $\mA$ is an i.i.d. Gaussian matrix \cite{LuL2020spectral,MondelliM2017spectral} or a subsampled random orthogonal matrix \cite{Dudeja2020spectral}. In our experiment, we consider the latter setting. Assume $m = \lfloor \alpha n \rfloor$ for some $\alpha > 1$. We can write
\begin{equation}\label{eq:A_ortho}
\mA = \begin{bmatrix}
\mI_n & \mat{0}_{n \times (m-n)}
\end{bmatrix} \mQ,
\end{equation}
where $\mQ \in \R^{m \times m}$ is a random orthogonal matrix drawn from the Haar distribution.

\begin{figure}
	\centering
	\includegraphics[width=0.46\linewidth]{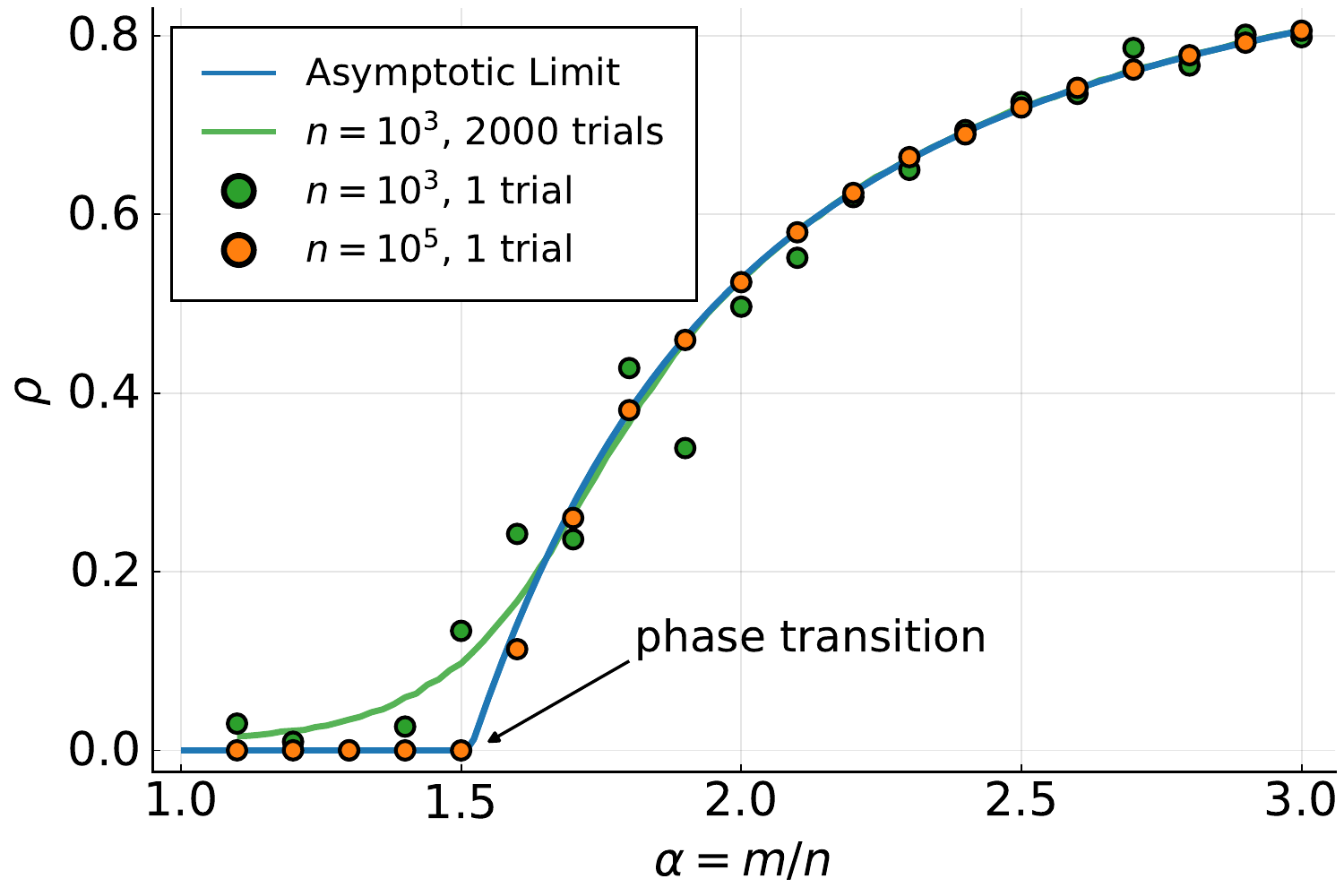}
	\caption{Simulating the spectral method given in \eref{D_mtx} and comparing the empirical results against the asymptotic predictions given in \cite{Dudeja2020spectral}. The result for $n = 10^3$ shows strong statistical fluctuations. This can be reduced by averaging over multiple independent trials, but the average curve still suffers from strong finite size effects, especially near the phase transition point. At $n = 10^5$, the match between the empirical results and the theoretical curve is nearly perfect in any (typical) trial.}
	\label{fig:spectral}
\end{figure}

We simulate the spectral method and compare its empirical performance with the asymptotic limit given in \cite{Dudeja2020spectral}. In our experiment, the measurement model is set to be $y_i = \tanh\big(\abs{\va_i\tran \vxi}\big)$. We compute the leading eigenvector $\vx_1$ by using the Krylov-Schur algorithm \cite{Stewart2002eigen}, which involves the repeated multiplication of $\mD$ with some vectors. With the forms of $\mD$ and $\mA$ given above, this algorithm can again be regarded as a special case of the general dynamics in \eref{dynamics}. We use the HD algorithm for the simulation and show the results in \fref{spectral} for two different matrix dimensions: $n = 10^3$ and $n = 10^5$. Observe that, at $n = 10^3$, there is still noticeable fluctuations between the actual performance of the spectral method (shown as green dots in the figure) and the theoretical prediction (the blue line). To get a better match, the standard practice is to do many independent trials (2000 in our experiment) and average the results. This gives us the green curve in the figure. Averaging can indeed reduce statistical fluctuations, but there are still strong finite size effects, especially near the phase transition point. This is a case where the capability of the proposed HD algorithm to handle large matrices becomes particularly attractive: when we increase the dimension to $n = 10^5$, the empirical results match the theoretical curve very closely in any (typical) trial, with no need for averaging over repeated simulations. In terms of runtime, it takes the HD algorithm less than 4 seconds on average to obtain an extremal eigenvalue/eigenvector of $\mD$ for $n = 10^5$.

%\subsection{Inference on Multilayer Neural Networks}
%
%\fbox{note:} cite ML AMP, tree AMP, ML-VAMP

\section{Main Results}
\label{sec:main}

\emph{Notation}: In what follows, $\ve_i$ denotes the $i$th natural basis vector, and $\mZ_i \bydef \mI - \ve_i \ve_i\tran$. For $i \le j$, we use $\mZ_{i:j}$ as a shorthand notation for $\prod_{i \le k \le j} \mZ_k$. The dimension of $\mZ_i$ and $\mZ_{i:j}$ is either $m \times m$ or $n \times n$, which will be made clear from the context. For any $\vv \in \R^n$, the ``slicing'' operation that takes a subset of $\vv$ is denoted by
\begin{equation}\label{eq:slicing}
\vv[i:j] \bydef [v_i, v_{i+1}, \ldots, v_j]\tran,
\end{equation}
where $1 \le i \le j \le n$. We use 
\[
\Ortho(n) \bydef \{\mM \in \R^{n \times n}: \mM \mM\tran = \mI_n\}
\]
to denote the set of $n\times n$ orthogonal matrices, and $\Unitary(n) \bydef \{\mM \in \C^{n \times n}: \mM \mM\herm = \mI_n\}$ its complex-valued counterpart. We will be mainly focusing on two real-valued random matrix ensembles: $\GE({m, n})$ represents the ensemble of $m \times n$ matrices with i.i.d. standard normal entries, and $\HE(n)$ represents the ensemble of random orthogonal matrices drawn from the Haar measure on $\Ortho(n)$. The generalizations to the complex-valued cases and other closely related ensembles will be discussed in \sref{complex}.

\subsection{Preliminaries}

The ensembles $\GE({m, n})$ and $\Ortho(n)$ share an important property: they are both \emph{invariant} with respect to multiplications by orthogonal matrices. For example, for any $\mG$ drawn from $\GE({m, n})$, it is easy to verify that
\begin{equation}\label{eq:Ginibre_inv}
\mG \sim \GE(m, n) \implies \mU \mG \mV \sim \GE(m, n),
\end{equation}
where $\mU\in \Ortho(m), \mV \in \Ortho(n)$ are any two deterministic or \emph{random} orthogonal matrices independent of $\mG$. 

Translation-invariant properties similar to \eref{Ginibre_inv} are actually what defines the Haar measure. We call a probability measure $\mu$ on $\Ortho(n)$ a Haar measure if 
\begin{equation}\label{eq:Haar_inv}
\mu(\mathcal{A}) = \mu(\mU \circ \mathcal{A}) = \mu(\mathcal{A} \circ \mU)
\end{equation}
for any measurable subset $\mathcal{A} \subset \Ortho(n)$ and any fixed $\mU \in \Ortho(n)$. Here, $\mU \circ \mathcal{A} \bydef \set{\mU \mV: \mV \in \mathcal{A}}$ and $\mathcal{A} \circ \mU$ is defined similarly. The classical Haar's theorem \cite{Haar1933measure,Meckes2019rmt} shows that there is one, and only one, translation-invariant probability measure in the sense of \eref{Haar_inv} on $\Ortho(n)$. In fact, the theorem holds in much greater generality. For example, it remains true for any compact Lie group, which includes $\Ortho(n)$ [and $\Unitary(n)$] as its special case. 

An additional property of $\Ortho(n)$, $\Unitary(n)$ (and compact Lie groups in general) is that left-invariance [the first equality in \eref{Haar_inv}] implies right-invariance (the second equality), and vice versa. This then allows us to have a simplified characterization of the Haar measure on $\Ortho(n)$. Specifically, to show that a random orthogonal matrix $\mQ \sim \HE(n)$, it is sufficient to verify that
\[
\mQ \eqd \mU \mQ
\]
for any fixed $\mU \in \Ortho(n)$, where $\eqd$ means that two random variables have the same distribution. We will use this convenient characterization in \sref{Haar}, when we establish the statistical equivalence between the proposed HD algorithm and the direct simulation of \eref{dynamics}.

Finally, we recall the construction of Householder reflectors \cite{Householder1958reflector,TrefethenBau1997numerical} from numerical linear algebra, as they will play important roles in our subsequent discussions. Given a vector $\vv \in \R^n$, how can we build an orthogonal matrix $\mH$ such that $\mH \vv = \norm{\vv} \ve_1$? This is exactly the problem addressed by Householder reflectors, defined here as
\begin{equation}\label{eq:Householder}
\mH(\vv) \bydef -\sign(v_1) \Big(\mI - 2\,  \frac{\vu \vu\tran}{\vu\tran\vu}\Big),
\end{equation}
where $\vu = \vv + \sign(v_1) \norm{\vv} \ve_1$, and $\sign(v_1) = 1$ if $v_1 \ge 0$ and $0$ otherwise. The choice of the sign in \eref{Householder} helps improve numerical stability (see \cite[Lecture 10]{TrefethenBau1997numerical}). 

By construction, $\mH(\vv)$ is a symmetric matrix whose eigenvalues are equal to $\pm 1$. It follows that $\mH(\vv) \in \Ortho(n)$. Moreover, we can verify from direct calculations that
\begin{equation}\label{eq:Householder_property}
\mH(\vv) \ve_1 = \vv/\norm{\vv} \quad \text{and} \quad \mH(\vv) \vv = \norm{\vv} \ve_1.
\end{equation}
Geometrically, $\mH(\vv)$ represents a reflection across the exterior (or interior) angle bisector of $\vv/\norm{\vv}$ and $\ve_1$. It is widely used in numerical linear algebra thanks to its low memory/computational costs. The matrix $\mH(\vv)$ itself can be efficiently represented with $\mathcal{O}(n)$ space, and matrix-vector multiplications involving $\mH(\vv)$ only require $\mathcal{O}(n)$ work.

For any $\vp \in \R^n$ and $1 \le k \le n$, we define a generalized Householder reflector as
\begin{equation}\label{eq:Householder_e}
\mH_{k}(\vp) \bydef \begin{bmatrix}
\mI_{k-1} &\\
& \mH(\vp[k:n])
\end{bmatrix},
\end{equation}
where $\mH(\cdot)$ is the reflector defined in \eref{Householder}, and $\vp[k:n]$ denotes a subvector obtained by removing the first $k-1$ elements of $\vp$. The construction in  \eref{Householder} requires that the reflecting vector $\vp[k:n]$ be nonzero. In order for \eref{Householder_e} to be always well-defined, we set $\mH_{k}(\vp) = \mI_n$ if $\vp[k:n] = \vzero$. Recall the notation $\mZ_{1:k}$ introduced at the beginning of the section. It is easy to verify that
\begin{equation}\label{eq:ortho}
\mZ_{1:k} \mH_{k}(\vp) \vp = \vzero,
\end{equation}
which means that the orthogonal transformation $\mH_{k}(\vp)$ can turn the last $n-k$ entries of $\vp$ to zero.
%and that
%\begin{equation}\label{eq:commute}
%\mZ_{r-1} \mH_n(\vp[r:n]) = \mH_n(\vp[r:n])\mZ_{r-1}.
%\end{equation}
We will use this property in the construction of the HD algorithm.

\subsection{Gaussian Random Matrices}

We start by considering the case where the random matrix $\mQ$ in the dynamics \eref{dynamics} has i.i.d. Gaussian entries, \emph{i.e.}, $\mQ \sim \GE(m, n)$. In addition, we shall always assume that $\mQ$ is independent of the initial condition $\set{\vx_1, \vx_{0}, \ldots, \vx_{1-d}}$.

Suppose that the first step of \eref{dynamics} is in the form of $\vx_2 = f_1(\mQ \vx_1, \vx_1, \ldots, \vx_{1-d})$, i.e., $\mM_1 = \mQ$. How do we simulate this step without generating the entire Gaussian matrix $\mQ$? This can be achieved by a simple observation:
%\begin{align}
%\mQ &\eqd \vg_1 \ve_1\tran + \mG \mZ_1\\
%	&\eqd (\vg_1 \ve_1\tran + \mG_1 \mZ_1) \mH_n(\vx_0)\\
%	&= \vu_1 \vv_1\tran + \mG_1 \mZ_1 \mH_n(\vx_0)
%\end{align}
\begin{equation}\label{eq:Gaussian_factorization}
\mQ \eqd \vg_1 \ve_1\tran + \mG_1 \mZ_1 \eqd (\vg_1 \ve_1\tran + \mG_1 \mZ_1) \mR_1 \sim \GE(m, n),
\end{equation}
where $\mZ_1 = \mI - \ve_1 \ve_1\tran$, $\mR_1 \bydef \mH_{1}(\vx_1)$ is a (generalized) Householder reflector defined in \eref{Householder_e}, $\vg_1 \sim \GE(m, 1)$ is a Gaussian vector, and $\mG_1 \sim \GE(m, n)$ is an independent Gaussian matrix. Here and subsequently, whenever we generate new random vectors and matrices, they are always independent of each other and of the $\sigma$-algebra generated by all the other random variables constructed up to that point. For example, $\vg_1$ and $\mG_1$ in \eref{Gaussian_factorization} are understood to be independent of the initial condition $\set{\vx_1, \vx_0, \ldots, \vx_{1-d}}$. In \eref{Gaussian_factorization}, the first equality (in distribution) is obvious, and the second equality is due to the translation invariance of the Ginibre ensemble. (Recall \eref{Ginibre_inv} and the fact that $\mR_1$ is an orthogonal matrix.) 

The new representation
\begin{equation}\label{eq:Gaussian_1}
\mQ^{(1)} = (\vg_1 \ve_1\tran + \mG_1 \mZ_1) \mR_1
\end{equation}
looks like a rather convoluted way of writing an i.i.d. Gaussian matrix, but it turns out to be the right choice for efficient simulations. To see this, we use the property of the Householder reflector [see \eref{Householder_property}] which gives us $\mR_1 \vx_1 = \mH_{1}(\vx_1) \vx_1 = \norm{\vx_1} \ve_1$ and thus $\mZ_1 \mR_1 \vx_1 = \vzero$. It follows that
\[
\mQ^{(1)} \vx_1 = \norm{\vx_1} \vg_1.
\]
Thus, to simulate the first step of the dynamics, we only need to generate a Gaussian vector $\vg_1$. The more expensive Gaussian matrix $\mG_1$ does not need to be revealed (yet), as it is invisible to $\vx_1$.

It is helpful to consider two more iterations to see how this idea can be applied recursively. Suppose that the second iteration takes the form of $\vx_3 = f_2(\mQ \vx_2, \vx_2, \ldots, \vx_{2-d})$. In general, $\vx_2$ will have a nonzero component in the space orthogonal to $\vx_1$, and thus the Gaussian matrix $\mG_1$ in \eref{Gaussian_1} is no longer invisible to $\vx_2$, meaning that $\mG_1 \mZ_1 \mR_1 \vx_2 \neq \vzero$. However, we can use the trick in \eref{Gaussian_factorization} again by writing
\begin{equation}\label{eq:G12}
\mG_1 \eqd (\vg_2 \ve_2\tran + \mG_2 \mZ_2) \mR_2 \sim \GE(m, n),
\end{equation}
where $\vg_2 \sim \GE(m, 1)$, $\mG_2 \sim \GE(m, n)$, and $\mR_2 \bydef \mH_2(\mR_1 \vx_2)$ is again a generalized Householder reflector in \eref{Householder_e}. The subscript in $\mH_2$ should not be overlooked, as it signifies the precise way the matrix is constructed. [Recall \eref{Householder_e} for the notation convention we use.]

Observe that $\mR_2$ commutes with $\mZ_1$. Substituting \eref{G12} into \eref{Gaussian_1} then allows us to write
\begin{equation}\label{eq:Gaussian_2}
\mQ^{(2)} = \vu_1 \vv_1\tran + \vu_2 \vv_2\tran + \mG_2 \mZ_{1:2} \mR_2 \mR_1 \sim \GE(m, n),
\end{equation}
where $\vu_1 \bydef \vg_1$, $\vu_2 \bydef \vg_2$, $\vv_1 \bydef \mR_1 \ve_1$, and $\vv_2 \bydef \mR_1 \mR_2 \ve_2$. Just like what happens in \eref{Gaussian_1}, there is again no need to explicitly generate the dense Gaussian matrix $\mG_2$ in \eref{Gaussian_2}. To see this, we note that $\mZ_{1:2} \mR_2 \mR_1 \vx_2 = \mZ_{1:2} \mH_2(\mR_1 \vx_2) \mR_1 \vx_2 = \vzero$, where the second equality is due to \eref{ortho}. It follows that
\[
\mQ^{(2)} \vx_2 =(\vv_1\tran \vx_2) \vu_1 + (\vv_2\tran \vx_2) \vu_2.
\]

So far we have only been considering the case where we access $\mQ$ from the right. For the third iteration, let us suppose that we access $\mQ$ from the left, i.e., $\vx_4 = f_3(\mQ\tran \vx_3, \vx_3, \ldots, \vx_{3-d})$. The idea is similar. Let
\begin{equation}\label{eq:G23}
\mG_2 = \mL_1 (\ve_1 \vg_3\tran + \mZ_1 \mG_3) \sim \GE(m, n),
\end{equation}
where $\mL_1 \bydef \mH_1(\vx_3)$, $\vg_3 \sim \GE(n, 1)$, and $\mG_3 \sim \GE(m, n)$. Substituting \eref{G23} into \eref{Gaussian_2} gives us
\[
\mQ^{(3)} = \textstyle\sum_{1 \le i \le 3} \vu_i \vv_i\tran + \mL_1 \mZ_1 \mG_3 \mZ_{1:2} \mR_2 \mR_1 \sim \GE(m, n),
\]
where $\vu_3 \bydef \mL_1 \ve_1$ and $\vv_3 \bydef \mR_1 \mR_2 \mZ_{1:2} \vg_3$. Moreover, $[\mQ^{(3)}]\tran \vx_3 = \sum_{i \le 3} (\vu_i\tran \vx_3) \vv_i$. 

The general idea should now be clear. Rather than fixing the entire Gaussian matrix in advance, we let the random choices gradually unfold as the iteration goes on, generating only the randomness that must be revealed at each step. Continuing this process for $T$ steps, we reach the HD algorithm for the Ginibre ensemble, summarized in Algorithm~\ref{alg:Gaussian}. Its memory and computational costs can be determined as follows.

\begin{algorithm}[t]
\caption{Simulating \eref{dynamics} on $\GE(m,n)$ using Householder Dice}\label{alg:Gaussian}
\begin{algorithmic}[1]
\Require The initial condition $\set{\vx_1, \vx_{0}, \ldots, \vx_{1-d}}$, and the number of iterations $T \le \min\set{m, n}$
\State Set $r = 0$, $\ell = 0$, $\mL_0 = \mI_m$, and $\mR_0= \mI_n$. 
\For {$t = 1, \ldots, T$}
	\If {$\mM_t = \mQ$}
		\State $r \gets r + 1$
		\State Generate $\vg_t \sim \GE(m, 1)$
		\State $\mR_r = \mH_{r}(\mR_{r-1}\ldots \mR_{1} \mR_0 \vx_t)$ \label{alg:gc1}
		\State $\vu_t = \mL_0 \mL_1 \ldots \mL_\ell  \mZ_{1:\ell}\vg_t$
		\State $\vv_t = \mR_0 \mR_1 \ldots \mR_r \ve_r$
		\State $\vy_t = \sum_{i \le t} (\vv_i\tran \vx_t) \vu_i$ \label{alg:gc2}
	\Else %\Comment{$\mM_t$ is the transpose of $\mQ$}
		\State $\ell \gets \ell + 1$
		\State Generate $\vg_t \sim \GE(n, 1)$
		\State $\mL_\ell = \mH_{\ell}(\mL_{\ell-1}\ldots \mL_{1} \mL_0 \vx_t)$ \label{alg:gc1a}
		\State $\vu_t = \mL_0 \mL_1 \ldots \mL_\ell \ve_\ell$
		\State $\vv_t = \mR_0 \mR_1 \ldots \mR_r  \mZ_{1:r}\vg_t$
		\State $\vy_t = \sum_{i \le t} (\vu_i\tran \vx_t) \vv_i$ \label{alg:gc2a}
	\EndIf
	\State $\vx_{t+1} = f_t(\vy_t, \vx_t, \vx_{t-1}, \ldots, \vx_{t-d})$
\EndFor
\end{algorithmic}
\end{algorithm}

%At the $t$-th iteration, let
%\begin{equation}\label{eq:bt}
%b_t = \begin{cases}
%1 &\text{if } \mM_t = \mQ\\
%0 &\text{if } \mM_t = \mQ\tran,
%\end{cases}, \quad r_t = \textstyle\sum_{1 \le i \le t} b_i, \quad \text{and}\quad \ell_t = t - r_t.
%\end{equation}
%Clearly, $r_t$ [resp. $\ell_t$] records the number of times we have used $\mQ$ (resp. $\mQ\tran$) in the first $t$ iterations of the dynamics. 

During its operation, Algorithm~\ref{alg:Gaussian} keeps track of $2T$ vectors $\set{\vu_t \in \R^m, \vv_t \in \R^n}_{t \le T}$ and $T$ Householder reflectors
\[
\set{\mL_i \in \R^{m\times m}}_{i \le \ell_T} \text{ and } \set{\mR_i \in \R^{n\times n}}_{i \le r_T},
\]
where $\ell_T$ (resp. $r_T$) records the number of times we have used $\mQ\tran$ (resp. $\mQ$) in the $T$ iterations of the dynamics. Clearly, $r_T +\ell_T = T$.  Thanks to the structures of the Householder reflectors in \eref{Householder}, the total memory footprint of Algorithm~\ref{alg:Gaussian} is $\mathcal{O}((m+n)T)$. At each iteration, computations mainly take place in lines \ref{alg:gc1}--\ref{alg:gc2} (or lines \ref{alg:gc1a}--\ref{alg:gc2a} if $\mM_t = \mQ\tran$). Since the matrices used there are always products of Householder reflectors, these steps require $\mathcal{O}((m+n)t)$ operations. As $t$ ranges from 1 to $T$, the computational complexity of Algorithm~\ref{alg:Gaussian} is thus $\mathcal{O}((m+n)T^2)$.

\begin{remark}
In line \ref{alg:gc1} and line \ref{alg:gc1a}, Algorithm~\ref{alg:Gaussian} recursively constructs two products of (generalized) Householder reflectors. Readers familiar with numerical linear algebra will recognize that this process is essentially the Householder algorithm for QR factorization \cite[Lecture 10]{TrefethenBau1997numerical}. Special data structures have been developed (see, e.g., \cite{SchreiberV1989storage}) to efficiently represent and operate on such products of reflectors.
\end{remark}

We can now exhibit the statistical equivalence of the HD algorithm and the direct simulation approach.

\begin{theorem}\label{thm:Gaussian}
Fix $T \le \min\set{m, n}$, and let $\set{\vx_t: 1-d \le t \le T+1}$ be a sequence of vectors generated by Algorithm~\ref{alg:Gaussian}. Let $\set{\widetilde{\vx}_t: 1-d \le t \le T+1}$ be another sequence obtained by the direct approach to simulating \eref{dynamics}, where we use the same initial condition (i.e. $\widetilde{\vx}_t = \vx_t$ for $1-d \le t \le 1$) but generate a full matrix $\mQ \sim \GE(m, n)$ in advance. The joint probability distribution of $\set{\vx_t}$ is equivalent to that of $\set{\widetilde{\vx}_t}$.
\end{theorem}
\begin{proof}
We start by describing the general structure of the algorithm. At the $t$-th iteration, the algorithm keeps the following representation of the matrix $\mQ$:
\begin{equation}\label{eq:Gaussian_t}
\mQ^{(t)} = \textstyle\sum_{i \le t} \vu_i \vv_i\tran + \underbrace{\mL_0 \mL_1 \ldots \mL_{\ell_t}}_{\text{Householder}} \mZ_{1:\ell_t} \mG_t \mZ_{1:r_t} \underbrace{\mR_{r_t} \ldots \mR_1 \mR_0}_{\text{Householder}},
\end{equation}
where $\mG_t \sim \GE(m, n)$ is a Gaussian matrix independent of the $\sigma$-algebra generated by all the other random variables constructed up to this point, and $\ell_t$ (resp. $r_t$) denotes the number of times we have used $\mQ\tran$ (resp. $\mQ$) in the first $t$ iterations of the dynamics. To lighten the notation, we will omit the subscript in the remainder of the proof and simply write them as $\ell$ and $r$.

The vectors $\set{\vu_i, \vv_i}$ and the Householder reflectors $\set{\mL_i}$, $\set{\mR_i}$ in \eref{Gaussian_t} are constructed recursively, as follows. We start with $\mL_0 = \mI_m$ and $\mR_0 = \mI_n$. At the $t$-th iteration (for $1 \le t \le T$), if $\mM_t = \mQ$ (i.e. if we need to compute $\mQ \vx_t$), we add a new Householder reflector
\begin{equation}\label{eq:Rt_g}
\mR_{r} = \mH_{r}(\mR_{r-1} \ldots \mR_1 \mR_0 \vx_t)
\end{equation}
and two new ``basis'' vectors
\[
\vu_t =  \mL_0 \mL_1 \ldots \mL_{\ell}  \mZ_{1:\ell}\vg_t \quad \text{and} \quad \vv_t = \mR_0 \mR_1 \ldots \mR_{r} \ve_{r},
\]
where $\vg_t \sim \GE(m, 1)$. The procedure for the case of $\mM_t = \mQ\tran$ is completely analogous: we add a new Householder reflector $\mL_{\ell}$ (on the left) and construct the basis vectors $\vu_t, \vv_t$ accordingly.

It is important to note that the Gaussian matrix $\mG_t$ in \eref{Gaussian_t} is never explicitly constructed in the algorithm. Assume without loss of generality that $\mM_t = \mQ$. Let $\vp = \mR_{r-1} \ldots \mR_1 \mR_0 \vx_t$. We then have
\[
\mZ_{1:r} \mR_{r} \ldots \mR_1 \mR_0 \vx_t = \mZ_{1:r} \mH_{r}(\vp) \vp = \vzero,
\]
where the second equality is due to \eref{ortho}. Consequently, $\mG_t$ remains invisible to $\vx_t$, and 
\[
\mQ^{(t)} \vx_t = \textstyle\sum_{i \le t} (\vv_i\tran \vx_{t}) \vu_i.
\]

To prove the assertion of the theorem, it suffices to show that, for all $1\le t \le T$, $\mQ^{(t)}$ has the correct distribution, namely $\mQ^{(t)} \sim \GE(m, n)$ and $\mQ^{(t)}$ is independent of the initial condition $\set{\vx_1, \vx_0, \ldots, \vx_{1-d}}$.  This is clearly true for $t = 1$, based on our discussions around \eref{Gaussian_1}. Now suppose that the condition on the distribution has been verified for $\mQ^{(t)}$ for some $t \ge 1$, . To go to $t+1$, we rewrite the Gaussian matrix $\mG_t$ in \eref{Gaussian_t} by using a decomposition similar to \eref{G12}. Specifically, if $\mM_t = \mQ$, we write
\begin{equation}\label{eq:Gtt1}
\mG_t \eqd (\vg_{t+1} \ve_{r+1}\tran + \mG_{t+1} \mZ_{r+1}) \mR_{r+1} \sim \GE(m, n),
\end{equation}
where $\vg_{t+1} \sim \GE(m, 1)$, $\mG_{t+1} \sim \GE(m, n)$, and $\mR_{r+1} \bydef \mH_{r+1}(\mR_r \ldots, \mR_1 \mR_0 \vx_{t+1})$. (The decomposition for the case where $\mM_t = \mQ\tran$ is completely analogous.)

That the new representation on the right-hand side of \eref{Gtt1} has the same distribution as $\mG_t$ is due to the translation-invariant property of the Ginibre ensemble [see \eref{Ginibre_inv}]. Substituting \eref{Gtt1} into \eref{Gaussian_t} allows us to conclude that the matrix
\begin{equation}\label{eq:G_t1a}
\textstyle\sum_{i \le t} \vu_i \vv_i\tran + \mL_0 \ldots \mL_\ell \mZ_{1:\ell}(\vg_{t+1} \ve_{r+1}\tran + \mG_{t+1} \mZ_{r+1}) \mR_{r+1} \mZ_{1:r} \mR_r \ldots \mR_0
\end{equation}
satisfies the required condition on its distribution. By construction, $\mR_{r+1}$ commutes with $\mZ_{1:r}$. [Recall \eref{Householder_e}.] This simple observation allows us to check that the matrix in \eref{G_t1a} is exactly $\mQ^{(t+1)}$. By induction on $t$ from 1 to $T$, we then complete the proof.
\end{proof}

\subsection{Haar-Distributed Random Orthogonal Matrices}
\label{sec:Haar}

We now turn to the case where $\mQ$ is a Haar-distributed random orthogonal matrix. The construction of the HD algorithm relies on the following factorization of the Haar measure on $\mathbb{O}(n)$.
\begin{lemma}\label{lemma:Haar_factorization}
Let $\vg \sim \GE(n, 1)$, $\mQ_{n-1} \sim \HE(n-1)$, and $\vv \in \R^n$, all of which are independent. Then
\begin{equation}\label{eq:Haar_factorization}
\mH_1(\vg) \begin{bmatrix}
1 &\\
& \mQ_{n-1}
\end{bmatrix} \mH_1(\vv) \sim \HE(n).
\end{equation}
\end{lemma}
\begin{proof}
Let $\mM$ denote the left-hand side of \eref{Haar_factorization}. It is sufficient to show that $\mM \eqd \mU \mM$ for any fixed $\mU \in \mathbb{O}(n)$. The statement of the lemma then follows from the fact that the Haar measure is the unique (left) translation-invariant measure on $\mathbb{O}(n)$.

For any nonzero vector $\vx \in \R^n$, we denote by $\mB(\vx) \in \R^{n \times (n-1)}$ the submatrix consisting of the last $n-1$ columns of $\mH_1(\vx)$. It is also useful to notice that the first column of $\mH_1(\vx)$ is $\vx/\norm{\vx}$. Thus, $\mH_1(\vx) = \big[\frac{\vx}{\norm{\vx}} \mid \mB(\vx)\big]$. The following observation is easy to verify. For any fixed $\mU \in \mathbb{O}(n)$, there exists some $\mR \in \mathcal{O}(n-1)$ such that
\[
\mU \mB(\vx) = \mB(\mU \vx) \mR.
\]
Applying this to $\mB(\vg)$ [in $\mH_1(\vg)$] then allows us to write
\[
\mU \mM =\mH_1(\mU \vg) \begin{bmatrix}
1 &\\
& \mR \mQ_{n-1}
\end{bmatrix} \mH_1(\vv),
\]
where $\mR$ is an orthogonal matrix independent of $\mQ_{n-1}$ and $\vv$. Since the joint distribution of $(\mU\vg, \mR \mQ_{n-1}, \vv)$ is equal to that of $(\vg, \mQ_{n-1}, \vv)$ in \eref{Haar_factorization}, we must have $\mM \eqd \mU \mM$.
\end{proof}

The HD algorithm exploits the factorization in \eref{Haar_factorization} to speed up the simulation of Haar random matrices. Before presenting the algorithm in its full generality, we first illustrate how it unfolds in the first two iterations of \eref{dynamics}. For simplicity, we assume that $\mM_1 = \mM_2 = \mQ$. For the first iteration, we use \eref{Haar_factorization} to write $\mQ$ as
\begin{equation}\label{eq:H_1}
\mQ^{(1)} =  \mL_1 \begin{bmatrix}
1 &\\
& \mQ_{n-1}
\end{bmatrix} \mR_1\sim \HE(n),
\end{equation}
where $\mR_1 = \mH_1(\vx_1)$, $\mL_1 = \mH_1(\vg_1)$, $\vg_1 \sim \GE(n, 1)$ and $\mQ_{n-1} \sim \HE(n-1)$. 
Using the property of Householder reflectors given in \eref{Householder_property}, we have
\[
\mQ^{(1)} \vx_1 = \norm{\vx_1} \mH_1(\vg_1) \ve_1 = \frac{\norm{\vx_1}}{\norm{\vg_1}}\,\vg_1.
\]
Notice that only a Gaussian vector $\vg_1$ is needed here, and that the matrix $\mQ_{n-1}$ is invisible.

To simulate the second iteration, we apply the factorization \eref{Haar_factorization} recursively to write $\mQ_{n-1}$ as
\begin{equation}\label{eq:H12}
\mQ_{n-1} = \mH_{1}(\vg_2[2:n]) \begin{bmatrix}
1 &\\
& \mQ_{n-2}
\end{bmatrix} \mH_{1}(\vp[2:n]) \sim \HE(n-1),
\end{equation}
where $\vg_2 \sim \GE(n, 1)$, $\mQ_{n-2} \sim \HE(n-2)$, and $\vp = \mR_1 \vx_2$. Substituting \eref{H12} into \eref{H_1} then gives us
\begin{equation}\label{eq:H_2}
\mQ^{(2)}  = \mL_1 \mL_2 \begin{bmatrix}
\mI_2 &\\
&\mQ_{n-2}
\end{bmatrix} \mR_2 \mR_1,
\end{equation}
where $\mL_2 = \mH_2(\vg_2)$ and $\mR_2 = \mH_2(\vp)$. By construction, the vector $\mR_2 \mR_1 \vx_2$ has nonzero entries only in the first two coordinates. It follows that
\[
\mQ^{(2)} \vx_2 = \mL_1 \mL_2 \mR_2 \mR_1 \vx_2,
\]
with $\mQ_{n-2}$ in \eref{H_2} remaining invisible. 

\begin{algorithm}[t]
\caption{Simulating \eref{dynamics} on $\HE(n)$ using Householder Dice}\label{alg:Haar}
\begin{algorithmic}[1]
\Require The initial condition $\set{\vx_1, \vx_{0}, \ldots, \vx_{1-d}}$, and the number of iterations $T \le n$
\State Set $\mL_0 = \mI_m$, and $\mR_0= \mI_n$. 
\For {$t = 1, \ldots, T$}
	\State Generate $\vg_t \sim \GE(n, 1)$
	\If {$\mM_t = \mQ$}
		\State $\vp_t = \mR_{t-1} \ldots \mR_1 \mR_0 \vx_t$ \label{alg:hc_1}
		\State $\mR_t = \mH_t(\vp_t)$
		\State $\mL_t = \mH_t(\vg_t)$
		\State $\vy_t = \mL_{1} \ldots \mL_t \mR_t \vp_t$ \label{alg:hc_2}
	\Else %\Comment{$\mM_t$ is the transpose of $\mQ$}
		\State $\vp_t = \mL_{t-1} \ldots \mL_1 \mL_0 \vx_t$ \label{alg:hc_3}
		\State $\mL_t = \mH_t(\vp_t)$
		\State $\mR_t = \mH_t(\vg_t)$
		\State $\vy_t = \mR_{1} \ldots \mR_t \mL_t \vp_t$ \label{alg:hc_4}
	\EndIf
	\State $\vx_{t+1} = f_t(\vy_t, \vx_t, \vx_{t-1}, \ldots, \vx_{t-d})$
\EndFor
\end{algorithmic}
\end{algorithm}

Continuing this process, we see a simple pattern emerging. We summarize it in Algorithm~\ref{alg:Haar}. In general, the algorithm recursively constructs two sequences of Householder reflectors $\set{\mL_t, \mR_t}_{t \le T}$, starting from $\mL_0 = \mR_0 = \mI_n$. At the $t$-th iteration, we first generate a new Gaussian vector $\vg \sim \GE(n, 1)$. Suppose $\mM_t = \mQ$, we compute
\begin{equation}\label{eq:pt}
\vp_t = \mR_{t-1} \ldots \mR_1 \mR_0 \vx_t
\end{equation}
and add two reflectors $\mR_t = \mH_t(\vp_t)$ and $\mL_t = \mH_t(\vg_t)$. The algorithm then proceeds to the next iteration by letting $\vx_{t+1} = f_t(\vy_t, \vx_t, \ldots, \vx_{t-d})$, where $\vy_t = \mL_{1} \ldots \mL_t \mR_t \vp_t$. The steps the algorithms takes if $\mM = \mQ\tran$ is exactly symmetric, with the roles of $\set{\mR_i}$ and $\set{\mL_i}$ switched. The computational and memory complexity of Algorithm~\ref{alg:Haar} is similar to that of Algorithm~\ref{alg:Gaussian}. Specifically, the Householder reflectors can be efficiently represented by the corresponding reflection vectors, at a cost of $\mathcal{O}(nT)$ space. At each iteration, the matrix-vector multiplications in lines \ref{alg:hc_1}, \ref{alg:hc_2}, \ref{alg:hc_3} and \ref{alg:hc_4} can all be implemented in $\mathcal{O}(nt)$ operations (thanks to the Householder structure). Therefore, the total computational complexity is $\mathcal{O}(nT^2)$.

Finally, we establish the statistical ``correctness'' of Algorithm~\ref{alg:Haar} in the following theorem.
\begin{theorem}\label{thm:Haar}
Fix $T \le n$, and let $\set{\vx_t: 1-d \le t \le T+1}$ be a sequence of vectors generated by Algorithm~\ref{alg:Haar}. Let $\set{\widetilde{\vx}_t: 1-d \le t \le T+1}$ be another sequence obtained by the direct approach to simulating \eref{dynamics}, where we use the same initial condition (i.e. $\widetilde{\vx}_t = \vx_t$ for $1-d \le t \le 1$) but generate a random orthogonal matrix $\mQ \sim \HE(n)$ in advance. The joint probability distribution of $\set{\vx_t}$ is equivalent to that of $\set{\widetilde{\vx}_t}$.
\end{theorem}
\begin{proof}
The proof is very similar to that of Theorem~\ref{thm:Gaussian}. At the $t$-th iteration, the algorithm has constructed a representation of the random orthogonal matrix $\mQ$ as
\begin{equation}\label{eq:H_t}
\mQ^{(t)} =  \mL_1 \mL_2 \ldots \mL_t \begin{bmatrix}
\mI_t &\\
&\mQ_{n-t}
\end{bmatrix} \mR_t \ldots \mR_2 \mR_1,
\end{equation}
where $\set{\mL_i, \mR_i}_{i \le t}$ is a collection of Householder reflectors, and $\mQ_{n-t} \sim \HE(n-t)$ is an $(n-t) \times (n-t)$ random orthogonal matrix independent of the $\sigma$-algebra generated by all the other random variables constructed up to this point. We shall have established the theorem if we prove the following two claims for $1 \le t \le T$: (a) $\mQ^{(t)} \sim \HE(n)$ and $\mQ^{(t)}$ is independent of the initial condition $\set{\vx_t}_{1-d \le t \le 1}$; (b) If $\mM_t = \mQ$ in \eref{dynamics}, then
\begin{equation}\label{eq:Qtxt}
\mQ^{(t)}\vx_t = \mL_{1} \ldots \mL_t \mR_t \vp_t,
\end{equation}
where $\vp_t$ is as defined in \eref{pt}. If $\mM_t = \mQ\tran$, then $[\mQ^{(t)}]\tran \vx_t = \mR_1 \mR_2 \ldots \mR_t \mL_t \ldots \mL_2 \mL_1 \vx_t$.

Claim (a) can be proved by induction. We have already established its correctness for $t = 1$. [See \eref{H_1}.]  To propagate the claim from iteration $t$ to $t+1$, we simply apply Lemma~\ref{lemma:Haar_factorization} to rewrite $\mQ_{n-t}$ in \eref{H_t} as
\[
\mQ_{n-t} \eqd \mH_1(\vg_{t+1}[t+1:n]) \begin{bmatrix}
1 &\\
& \mQ_{n-t-1}
\end{bmatrix} \mH_1(\vp_{t+1}[t+1:n]) \sim \HE(n-t),
\]
where $\vg_{t+1} \sim \GE(n, 1)$, $\mQ_{n-t-1} \sim \HE(n-t-1)$, and $\vp_{t+1} = \mR_t \ldots \mR_2 \mR_1 \vx_{t+1}$. (This is for the case of $\mM_{t+1} = \mQ$, but the treatment for the case of $\mM_{t+1} = \mQ\tran$ is completely analogous.) Substituting this equivalent representation into \eref{H_t} gives us $\mQ^{(t+1)}$. 

To establish Claim (b), we again assume without loss of generality that $\mM_t = \mQ$. By the definition of $\vp_t$ in \eref{pt} and that of $\mR_t$, we have
\[
\mQ^{(t)} \vx_t =  \mL_1 \mL_2 \ldots \mL_t \begin{bmatrix}
\mI_t &\\
&\mQ_{n-t}
\end{bmatrix} \mH_t(\vp_t)\vp_t.
\]
Using \eref{ortho}, we can then verify the expression given in \eref{Qtxt}.
\end{proof}

\subsection{Other Random Matrix Ensembles}
\label{sec:complex}

The Gaussian and Haar ensembles studied above can serve as building blocks for simulating other related random matrix ensembles. For example, consider the classical Gaussian orthogonal ensemble (GOE). A symmetric $n\times n$ matrix $G$ is drawn from $\text{GOE}(n)$ if $\{G_{ij}\}_{1\le i \le j\le n}$ are independent random variables, with $G_{ii} \sim \mathcal{N}(0, 2)$ and $G_{ij}\sim \mathcal{N}(0, 1)$ for $i < j$. Clearly,
\[
\mQ \sim \GE(n, n) \implies \frac{1}{\sqrt{2}}(\mQ + \mQ\tran) \sim \text{GOE}(n).
\]
It follows that a single matrix-vector multiplication involving $\mG \sim \text{GOE}(n)$ can be simulated via two matrix-vector multiplications involving a nonsymmetric Gaussian matrix, i.e.,
\[
\vy = \mG \vx \implies \widehat{\vy}= \mQ \vx \text{ and } \vy = (\mQ\tran \vx + \widehat{\vy})/\sqrt{2}.
\]

As a second example, we consider random matrices in the form of 
\begin{equation}\label{eq:inv}
\mQ = \mU \mSigma \mV,
\end{equation}
where $\mU \sim \HE(m)$ and $\mV \sim \HE(n)$ are two independent random orthogonal matrices, and $\mSigma \in \R^{m \times n}$ is a rectangular diagonal matrix independent of $\mU, \mV$. Matrices like these often appear in the study of free probability theory \cite{MingoSpeicher2017free}. They are also used as a convenient model for matrices whose singular vectors are \emph{generic} \cite{OpperCW2016ising,RanganSF2019vamp,Fan2020amp}. Strictly speaking, Theorem~\ref{thm:Haar} only applies to the case where the dynamics operates on a single random orthogonal matrix. However, it is obvious from the proof that the idea applies to more general dynamics involving a finite number of independent random orthogonal matrices. Thus, Algorithm~\ref{alg:Haar} can be easily adapted to handle the matrix ensemble given in  \eref{inv}.

Finally, we note that the constructions of the HD algorithm can be generalized to the complex-valued cases, with the random matrices drawn from the complex Ginibre ensemble, the Haar ensemble on the unitary group $\mathbb{U}(n)$, and the Gaussian unitary ensemble, respectively. We avoid repetitions, as most changes in such generalizations are straightforward (such as replacing $\mM\tran$ by $\mM\herm$). In what follows, we only present the formula for a complex version of the Householder reflector, as it might be less well-known.

Let $\vv \in \C^n$ be a nonzero vector. Write $v_1 / \norm{\vv} = r e^{i \theta}$, where $r$ is a nonnegative real number. (When $v_1 = 0$, we have $r = 0$ and set $\theta=0$.) We define a unitary reflector  \cite[pp. 48--49]{Wilkinson1988algebraic} as
\begin{equation}\label{eq:Householder_c}
\mH(\vv) = (-e^{-i \theta}) \Big[\mI_n - \frac{(\vv/\norm{\vv} + e^{i \theta} \ve_1)(\vv/\norm{\vv} + e^{i \theta} \ve_1)\herm}{1+r}\Big].
\end{equation}
It is easy to check that $\mH(\vv)$ is a unitary matrix such that $\mH(\vv) \vv = \norm{\vv} \ve_1$ and $\mH\herm(\vv) \ve_1 = \vv / \norm{\vv}$. Moreover, if $\vv$ is real, then \eref{Householder_c} reduces to the Householder reflector given in \eref{Householder}. 

\section{Conclusion}
\label{sec:conclusion}

We proposed a new algorithm called Householder Dice for simulating dynamics on several dense random matrix ensembles with translation-invariant properties. Rather than fixing the entire random matrix in advance, the new algorithm is matrix-free, generating only the randomness that must be revealed at any given step of the dynamics. The name of the algorithm highlights the central role played by an adaptive and recursive construction of (random) Householder reflectors. These orthogonal transformations exploit the group symmetry of the matrix ensembles, while simultaneously maintaining the statistical correlations induced by the dynamics. Numerical results demonstrate the promise of the HD algorithm as a new computational tool in the study of high-dimensional random systems.

%One of the pleasures in working with high-dimensional probability theory is to see simulation results match with theoretical predictions.
%
%principle of deferred decision
%
%The theorem is true in much more generality (in particular, any compact Lie group has a Haar probability measure),
%
%Some Examples and Motivations
%Before discussing how to generate random matri- ces it is helpful to give a few examples that show how they appear in the applications.
%
%hence explicit formation and storage of U are not required. Only the ability to form the matrix-vector product AHw and a rank-one update to A are required.

\bibliographystyle{ieeetr}
\bibliography{refs}

\end{document}